\documentclass[reprint,amsmath,amsthm,amssymb,aps,pra]{revtex4-1}
\usepackage[margin=1.00in]{geometry}
\usepackage{mathtools}
\usepackage{amsmath,amsfonts,amssymb,amsthm,url,verbatim}
\usepackage{graphicx,subcaption}
\usepackage[toc,page]{appendix}
\usepackage{enumitem}
\usepackage[ pdftex, plainpages = false, pdfpagelabels,
                 pdfpagelayout = useoutlines,
                 bookmarks,
                 bookmarksopen = true,
                 bookmarksnumbered = true,
                 breaklinks = true,
                 linktocpage=all,
                 pagebackref=false,
                 colorlinks = true,
                 linkcolor = black,
                 urlcolor  = black,
                 citecolor = black,
                 anchorcolor = green,
                 hyperindex = true,
                 hyperfigures
                 ]{hyperref}
\usepackage[usenames, dvipsnames]{xcolor}
\usepackage{xifthen} 

\DeclarePairedDelimiter\bra{\langle}{\rvert}

\DeclarePairedDelimiter\ket{\lvert}{\rangle}
\DeclarePairedDelimiter\kket{\lvert}{\rangle\!\rangle}

\DeclarePairedDelimiterX\braket[2]{\langle}{\rangle}{#1 \delimsize\vert #2}
\DeclarePairedDelimiterX\bbraket[2]{\langle\!\langle}{\rangle\!\rangle}{#1 \delimsize\| #2}
\DeclarePairedDelimiterX\cbraket[2]{(\!(}{)\!)}{#1 \delimsize\| #2}
\DeclarePairedDelimiterX\ketbra[2]{\lvert}{\rvert}{#1 \delimsize\rangle\!\langle #2}
\DeclarePairedDelimiterX\kketbra[2]{\|}{\|}{#1 \delimsize\rangle\!\rangle\!\langle\!\langle #2}
\DeclarePairedDelimiterX\cketbra[2]{\|}{\|}{#1 \delimsize)\!)\!(\!( #2}
\DeclarePairedDelimiterX\inner[2]{\langle}{\rangle}{#1,#2}

\def\tr{{\rm tr}\,}

\newcommand{\arxiv}[2][]{\ifthenelse{\isempty{#1}}{\href{http://arxiv.org/abs/#2}{{\tt arXiv:\allowbreak{}#2}}} {\href{http://arxiv.org/abs/#2}{{\tt arXiv:\allowbreak{}#2 [#1]}}}}

\newcommand{\booktitle}{\textsl}
\newcommand{\hrefdoi}[2]{\href{https://dx.doi.org/#1}{#2}}

\newtheorem{theorem}{Theorem}
\newtheorem{lemma}{Lemma}

\newtheorem{conjecture}{Conjecture}
\newtheorem{prop}{Proposition}
\newtheorem{remark}{\normalfont\textit{Remark}}
\begin{document}
\title{L\"uders Channels and the Existence of Symmetric Informationally Complete Measurements}
\author{John B.\ DeBrota}
\author{Blake C.\ Stacey}
\affiliation{University of Massachusetts Boston, Morrissey Boulevard, Boston MA 02125, USA}
\date{\today}
\begin{abstract}
  The L\"uders rule provides a way to define a quantum channel given a
  quantum measurement. Using this construction, we establish an
  if-and-only-if condition for the existence of a $d$-dimensional
  Symmetric Informationally Complete quantum measurement (a SIC) in
  terms of a particular depolarizing channel. Moreover, the channel in
  question satisfies two entropic optimality criteria.
\end{abstract}
\maketitle

\section{Introduction}

A \textit{minimal informationally complete} (MIC) quantum measurement
for a $d$ dimensional Hilbert space $\mathcal{H}_d$ is a set of
linearly independent positive semidefinite operators $\{E_i\}$,
$i=1,\ldots,d^2$, which sum to the identity~\cite{DeBrota:2018a,
DeBrota:2018b}. If every element in a MIC is proportional to a rank-$n$ projector, we say the MIC itself is \textit{rank-n}. If the Hilbert--Schmidt inner products $\tr E_iE_j$ equal one constant for all $i\neq j$ and another constant when $i=j$, we say the MIC is \textit{equiangular}. A \textit{symmetric informationally complete} (SIC)
quantum measurement is a rank-1 equiangular MIC~\cite{Zauner:1999,
Renes:2004, Scott:2010, Fuchs:2017}. When a SIC $\{H_i\}$ exists, one can show
that $H_i=\frac{1}{d}\Pi_i$ where $\Pi_i$ are rank-1 projectors and
that
\begin{equation}
    \tr H_iH_j=\frac{d\delta_{ij}+1}{d^2(d+1)}\;.
\end{equation}

The theory of quantum channels provides a means to discuss the fully general way in which quantum states may be transformed. A standard result \cite{Nielsen:2010} has it that a quantum channel $\mathcal{E}$ may always be specified by a set of operators $\{A_i\}$, called \textit{Kraus operators}, such that for a quantum state $\rho$,
\begin{equation}
    \mathcal{E}(\rho)=\sum_iA_i\rho A_i^\dag\;.
\end{equation}

Consider a physicist Alice who is preparing to send a quantum system
through a channel that she models by $\mathcal{E}$. Alice initially
describes her quantum system by assigning to it a density matrix
$\rho$. The state $\mathcal{E}(\rho)$ encodes Alice's expectations for
measurements that can potentially be performed after the system is
sent through the channel. More specifically, let Alice's channel be a
\emph{L\"uders MIC channel} (LMC) associated with the MIC $\{E_i\}$, which may be understood in the following way. Alice plans to apply the
MIC $\{E_i\}$, and upon obtaining the result of that measurement, invoke the L\"uders rule~\cite{Barnum:2000, Busch:2009} to obtain a new state for her system,
\begin{equation}
    \rho'_i := \frac{1}{\tr \rho E_i} \sqrt{E_i} \rho \sqrt{E_i}\;,
\end{equation}
where we have introduced the \emph{principal Kraus operators}
$\{\sqrt{E_i}\}$, the unique positive semidefinite square roots of the MIC elements. Before applying her MIC, Alice can write the
\emph{post-channel state}
\begin{equation}
  \mathcal{E}(\rho) := \sum_i p(E_i) \rho'_i
  = \sum_i \sqrt{E_i} \rho \sqrt{E_i}\;,
\end{equation}
which is a weighted average of the states from which Alice plans to
select the actual state she will ascribe to the system after making
the measurement. (For more on the broader conceptual context of this
operation, see~\cite{Fuchs:2012, Stacey:2019}.)

LMCs are a proper subset of all quantum channels as many valid channels are unrelated to a MIC and do not admit a representation in terms of principal Kraus operators. For example, a unitary channel is not an LMC.

Throughout this paper, we will frequently use the fact that any MIC element $E_i$ is proportional to a density matrix $E_i:=e_i\rho_i$, where we call the proportionality constants $\{e_i\}$ the \textit{weights} of the MIC. Because the $\{E_i\}$ sum to the identity, the weights $\{e_i\}$ sum to the trace of the identity, which is just the dimension $d$. 

We refer to the LMC obtained from a SIC as the \textit{SIC channel}
$\mathcal{E}_{\rm SIC}$.  We may characterize the SIC channel in any
dimension in which a SIC exists using the convenient notion of a \textit{dual basis}. Given a basis for a vector space, any vector in that space is uniquely identified by its inner products with the basis elements. These inner products are the coefficients in the expansion of the vector over the elements of the dual basis. Likewise, the inner products with the elements of the dual basis are the coefficients in the expansion over the original basis. A consequence of this is that, if $\{H_i\}$ denotes the original basis and $\{\widetilde{H}_j\}$ denotes its dual basis, then
\begin{equation}
    \tr H_i \widetilde{H}_j = \delta_{ij}.
\end{equation}
It follows that if $\{H_i\}$ is a SIC, then the dual basis is given by
\begin{equation}
\widetilde{H}_j=(d+1)\Pi_j-I\;,
\end{equation}
so we may write an
operator $X\in\mathcal{L}(\mathcal{H}_d)$ in the SIC basis as
\begin{equation}
    \begin{split}
  X &= \sum_j(\tr X\widetilde{H}_j)H_j\\
    &= (d+1)\sum_j(\tr X\Pi_j)H_j - (\tr X)I\;.
\end{split}
\end{equation}
Noting that $(\tr X\Pi_j)H_j=\frac{1}{d}\Pi_jX\Pi_j$ and that
$\sqrt{H_j}=\frac{1}{\sqrt{d}}\Pi_j$, we obtain
\begin{equation}\label{sicchannel}
  \mathcal{E}_{\rm SIC}(X) = \frac{1}{d}\sum_j\Pi_jX\Pi_j
  = \frac{(\tr X)I+X}{d+1}\;,
\end{equation}
and so for any quantum state $\rho$,
\begin{equation}
  \mathcal{E}_{\rm SIC}(\rho) = \frac{I+\rho}{d+1}\;.
\end{equation}

Going forward, given an LMC $\mathcal{E}$ and input state $\rho$, let $\lambda$ denote the eigenvalue spectrum of the post-channel state $\mathcal{E}(\rho)$ and $\lambda_{\rm max}$ denote the maximum eigenvalue of this state. We use the notation $\overline{f(\ketbra{\psi}{\psi})}$ to denote the average value of the function $f(\ketbra{\psi}{\psi})$ over all pure state inputs $\ketbra{\psi}{\psi}$ with respect to the Haar measure. We now prove a lemma applicable to arbitrary LMCs upon which our
later results rely. 
\begin{lemma}\label{lemma}
    Let $\mathcal{E}$ be an LMC. Then
    \begin{equation}
        \overline{\lambda}_{\rm max}\geq \frac{2}{d+1}\;.
    \end{equation}
\end{lemma}
\begin{proof}
    For an arbitrary pure state $\rho=\ketbra{\psi}{\psi}$,
  \begin{equation}
    \mathcal{E}(\ketbra{\psi}{\psi})
    = \sum_i\sqrt{E_i}\ketbra{\psi}{\psi}\sqrt{E_i}\;.
  \end{equation}
  We may lower bound $\lambda_{\rm max}$ given
  such an input as follows:
  \begin{equation}
    \lambda_{\rm max}
    \geq
    \bra{\psi}\mathcal{E}(\ketbra{\psi}{\psi})\ket{\psi}
    = \sum_i|\bra{\psi}\sqrt{E_i}\ket{\psi}|^2\;.
    \end{equation}
If we now average over all pure states with the Haar measure, we will 
  produce a generic lower bound on the average maximal eigenvalue of
  the post-channel state:
\begin{equation}\label{lowerbndavg}
      \overline{\lambda}_{\rm max}
      \geq\sum_i\int_{\mathcal{H}}|\bra{\psi}\sqrt{E_i}\ket{\psi}|^2
      d\Omega_{\psi}\;.
         \end{equation}
         We can evaluate this integral using a known property of the Haar measure \cite{Renes:2004}. Integrating a tensor power over pure state projectors gives a result proportional to the projector $P_{\rm sym}$ onto the symmetric subspace of $\mathcal{H}_d\otimes\mathcal{H}_d$. Consequently,
         \begin{equation}\label{sumsym}
             \overline{\lambda}_{\rm max}\geq \frac{2}{d(d+1)}\sum_i\tr\!\!\left[\left(\sqrt{E_i}
             \otimes \sqrt{E_i}\right)P_{\rm sym}\right].
         \end{equation}
         Using the fact that 
         \begin{equation}
             P_{\rm sym}=\frac{1}{2}\left(I\otimes I+\sum_{kl}\ketbra{l}{k}\otimes\ketbra{k}{l}\right)\;,
         \end{equation}
         one may verify that
         \begin{equation}
             \tr\!\!\left[\left(\sqrt{E_i}\otimes \sqrt{E_i}\right)P_{\rm sym}\right]=\frac{1}{2}\left(\tr\sqrt{E_i}\right)^2 +\frac{e_i}{2}\;.
         \end{equation}
         Thus equation \eqref{sumsym} becomes 
         \begin{equation}
             \overline{\lambda}_{\rm max}\geq \frac{1}{d(d+1)}\left(\sum_i\left(\tr\sqrt{E_i}\right)^2+\sum_ie_i\right)\;.
         \end{equation}
         Now since $\tr\sqrt{E_i}=\sqrt{e_i}\tr\sqrt{\rho_i}$ and $\tr\sqrt{\rho_i}\geq1$,  
         \begin{equation}
             \overline{\lambda}_{\rm max}\geq\frac{2}{d(d+1)}\sum_ie_i=\frac{2}{d+1}\;.
         \end{equation}

\end{proof}
The following theorem reveals that the SIC channel's action is unique to SICs among LMCs.
\begin{theorem}\label{existenceproof}
  A SIC exists in dimension $d$ iff there is an LMC with action
  $\mathcal{E}(\rho)=\frac{I+\rho}{d+1}$ for all $\rho$.
\end{theorem}
\begin{proof}
  If a SIC exists, take $\mathcal{E}=\mathcal{E}_{\rm SIC}$. For the other direction, we will first demonstrate that the MIC which gives rise to this LMC must be rank-$1$. Having established this, we will be able to see that the unitaries relating different Kraus operators for this LMC are directly related to the MIC weights. This will allow us to show that the principal Kraus operators have the Gram matrix of a SIC and must form a SIC themselves.  

  For any pure state input, 
  \begin{equation}
      \lambda=\left(\frac{2}{d+1},\frac{1}{d+1},\ldots,\frac{1}{d+1}\right)\;.
  \end{equation}
  Thus the lower bound in Lemma \ref{lemma} is saturated. This can only occur when $\tr\sqrt{\rho_i}=1$ for all $i$ which implies that the MIC is rank-1.

  In \hyperref[quasiSIC]{Appendix A} we define and construct the \textit{quasi-SICs}, that is, sets of Hermitian, but not necessarily postive semidefinite, matrices $\{Q_i\}$ which have the same Hilbert--Schmidt inner products as SIC projectors, and we demonstrate that they furnish a Hermitian basis of constant-trace Kraus operators $A_i$ which give the same action as $\mathcal{E}$. Any other set of Kraus operators
  with the same effect will be related to this set by a unitary
  remixing, and since $\mathcal{E}$ is an LMC, we must have
  \begin{equation}\label{remix}
    \sqrt{E_i}=\sum_j[U]_{ij}A_j\;.
  \end{equation}
  Since we know the MIC is rank-$1$, we can trace both sides of
  this expression to obtain the identity
  $\sqrt{de_i}=\sum_j[U]_{ij}$. Furthermore, since the $A_j$ form a Hermitian basis, one may see that every element of $U$ must be real, rendering $U$ an orthogonal matrix. Then
  \begin{equation}
    \begin{split}
      \tr \!\left(\sqrt{E_i}\sqrt{E_j}\right)
      &= \frac{1}{d}\sum_{k,l}[U]_{ik}[U]_{jl}\tr Q_k Q_l\\
     &= \frac{1}{d}\sum_{k,l}[U]_{ik}[U]_{jl}\frac{d\delta_{kl}+1}{d+1}\\
            &=\frac{d\delta_{ij}+d\sqrt{e_ie_j}}{d(d+1)}\;.
    \end{split}
  \end{equation}
  When $i=j$, we have
  \begin{equation}
    e_i=\frac{1+e_i}{d+1}\implies e_i=\frac{1}{d}\;,
  \end{equation}
  and so
  \begin{equation}\label{sqrtgram}
    \tr \!\left(\sqrt{E_i}\sqrt{E_j}\right)=\frac{d\delta_{ij}+1}{d(d+1)}\;,
  \end{equation}
  and because $E_i$ is rank-$1$,
  \begin{equation}
    \tr(E_iE_j)=\frac{d\delta_{ij}+1}{d^2(d+1)}\;,
  \end{equation}
  that is, the MIC is a SIC.
\end{proof}

\section{Depolarizing L\"uders MIC Channels}
The SIC channel falls within a class of channels called
\textit{depolarizing channels}~\cite{Wilde:2017}. A depolarizing channel is a channel
\begin{equation}
  \mathcal{E}_\alpha(\rho) = \alpha\rho+\frac{1-\alpha}{d}I\;,\qquad
  \frac{-1}{d^2-1}\leq\alpha\leq1\;.
\end{equation}
The SIC channel corresponds to $\alpha=\frac{1}{d+1}$. One might wish
to know when an LMC is a depolarizing channel. From Theorem
\ref{existenceproof}, we know the only LMC with $\alpha=\frac{1}{d+1}$
is the SIC channel. What range of $\alpha$ are achievable by LMCs?

The answer to this question is any $\frac{1}{d+1}\leq\alpha<1$. To see
this, note that the eigenvalue spectrum for a depolarizing channel given a pure state input is
\begin{equation}\label{depolarizingspectrum}
  \lambda\left(\mathcal{E}_\alpha(\ketbra{\psi}{\psi})\right)
  = \left(\alpha+\frac{1-\alpha}{d},\frac{1-\alpha}{d},\ldots,
  \frac{1-\alpha}{d}\right)\;.
\end{equation} 
Recall the lower bound on the average
maximal eigenvalue for any LMC given a pure state input from Lemma \ref{lemma} is $\frac{2}{d+1}$. As the spectrum for a depolarizing channel is constant for pure state inputs, the lower bound on the average is the lower bound for any pure state input. If $\lambda_{\rm
  max} = \frac{1-\alpha}{d}$, then $\frac{1-\alpha}{d} \geq
\alpha+\frac{1-\alpha}{d} \implies \alpha\leq0$. The more negative
$\alpha$ is, the larger the maximal eigenvalue would be, so the
largest it can get is when $\alpha = \frac{-1}{d^2-1}$, in which case
$\lambda_{\rm max} = \frac{d}{d^2-1}<\frac{2}{d+1}$. So, $\lambda_{\rm
  max} = \alpha+\frac{1-\alpha}{d} \geq \frac{2}{d+1} \implies \alpha
\geq \frac{1}{d+1}$. When $\alpha=1$, the channel is the identity
channel, in other words, not depolarizing at all. It is easy to prove
that were this to be implemented by an LMC, it would require
$\sqrt{E_i}=\frac{1}{d}I$, but this does not lead to a linearly
independent set and is not a MIC. If a SIC exists, however, a
depolarizing LMC exists for any $\frac{1}{d+1} \leq \alpha < 1$, as
the next proposition shows.
\begin{prop}
    Suppose a SIC exists in dimension $d$. For a nonzero $\beta\in\left[\frac{-1}{d-1},1\right]$ satisfying
  \begin{equation}
    \alpha=1-\frac{\left(\sqrt{1-\beta+d\beta}-\sqrt{1-\beta}\right)^2}{d+1}\;,
  \end{equation}
  or, equivalently,
  \begin{equation}
      \begin{split}
          \beta = \frac{1}{d^2}& \Bigr ( (d-2)(d+1)(1-\alpha)\\
          & +2\sqrt{(d+1)(1-\alpha)(1-\alpha+d^2\alpha)} \Bigr )\;,
   \end{split}
  \end{equation}
  the MIC
  \begin{equation}\label{sicmix}
    E_i=\frac{\beta}{d}\Pi_i+\frac{1-\beta}{d^2}I\;,
  \end{equation}
  where $\Pi_i$ is a SIC projector, gives rise to the LMC
  \begin{equation}\label{depolarizingLMC}
    \mathcal{E}_\alpha(\rho) =
    \alpha\rho+\frac{1-\alpha}{d}I\;,\qquad
    \frac{1}{d+1}\leq\alpha<1\;.
  \end{equation}
\end{prop}
\noindent One may check that the principal Kraus operators associated
with the MIC elements \eqref{sicmix} are given by
\begin{equation}\label{principalkraus}
  A_i=\frac{\sqrt{1-\beta+d\beta}-\sqrt{1-\beta}}{d}\Pi_i
  + \frac{\sqrt{1-\beta}}{d}I\;,
\end{equation}
and then a routine calculation and the characterization of the SIC
channel from Theorem \ref{existenceproof} confirms the claim of the
proposition.

\begin{remark}
When $\beta = 1$, the MIC $\{E_i\}$ is the original SIC, whereas when
$\beta$ equals its minimum allowed value $-1/(d-1)$, it is the
rank-$(d-1)$ equiangular MIC
\begin{equation}\label{minbeta}
    E_i = \frac{1}{d(d-1)}(-\Pi_i+I)\;,
\end{equation}
indirectly noted in prior work~\cite{DeBrota:2018b, Zhu:2016,
  DeBrota:2017} for extremizing a nonclassicality measure based on
negativity of quasi-probability.
\end{remark}

Do any LMCs give rise to depolarizing channels in dimensions where one does not have access to a SIC? If we
replace the SIC projector in equation \eqref{principalkraus} with a
quasi-SIC, we may form Kraus operators effecting the same depolarizing channel,
\begin{equation}\label{quasikraus}
  K_i=\frac{\sqrt{1-\beta+d\beta}-\sqrt{1-\beta}}{d}Q_i
  + \frac{\sqrt{1-\beta}}{d}I\;.
\end{equation}
From \eqref{qsicsquaresum}, one can check that these will square to a valid MIC. For arbitrary $\beta$, however, $K_i$ may fail to be positive semidefinite and would therefore not be a principal Kraus operator. From the definition of a quasi-SIC, one sees that the eigenvalues of $Q_i$ are bounded below by $-1$. Even in this worst case, one can easily derive that $K_i$ will be positive semidefinite for any nonzero $\beta\leq\frac{3}{d+3}$. This range of $\beta$ entitles any
$\alpha\geq\frac{d^2-d-1}{d^2-1}$. (The minimal $\alpha$ is obtained from the most negative $\beta$.) When $d=2$, this minimal $\alpha$ matches the lower bound achieved by the SIC channel because every quasi-SIC is a SIC in this dimension, but for all $d>2$ the inequality is strict and monotonically increases with dimension. In practice, the minimal eigenvalue among all of the quasi-SIC operators one constructs will be significantly larger than
$-1$, and so, depending on how close to a SIC one can make their quasi-SIC, one should be able to get significantly closer to the SIC bound than the $\alpha$ we have derived.

Fully classifying the MICs giving depolarizing LMCs for particular values of $\alpha>\frac{1}{d+1}$ appears to be a difficult problem; it is not clear what properties these MICs must satisfy. For example, squaring the $K_i$ operators from equation \eqref{quasikraus} results in MICs which are dependent on one's quasi-SIC implementation and need not be equiangular as the family in equation \eqref{sicmix} was. All principal Kraus operators which give rise to a depolarizing channel with a given $\beta$ (and corresponding $\alpha$) will be related to the operators \eqref{quasikraus} by way of a unitary remixing satisfying
\begin{equation}\label{remix2}
  \sqrt{E_i}=\sum_j[U]_{ij}K_j\;
\end{equation}
for some unitary $U$. As in the proof of Theorem \ref{existenceproof}, all the elements of the unitary must be real and so it is actually an orthogonal matrix. We have not been able to identify any further necessary characteristics of the $U$ in the completely general case, but the following notable restriction yielded further structure.
A MIC is \emph{unbiased} if the traces of all the elements are
equal, that is, if $e_i=\frac{1}{d}$ for all $i$. MICs in this class have the property that their measurement outcome probabilities for
the ``garbage state'' $\frac{1}{d}I$ input is the flat probability distribution over $d^2$ outcomes. From the standpoint of \cite{DeBrota:2018a}, this means they preserve the intution that the state $\frac{1}{d}I$ should correspond to a prior with complete outcome indifference in a reference process scenario and accordingly warrant special attention. If we demand that $\{E_i\}$ be unbiased, then
it is necessary, but not sufficient, that the orthogonal matrix remixing
\eqref{quasikraus} be doubly quasistochastic 
(see \hyperref[dblyquasi]{Appendix B}). 
\section{Entropic Optimality}
One way to evaluate the performance of a quantum channel is by using
measures based on von Neumann entropy,
\begin{equation}
    S(\rho)=-\tr\rho\log\rho\;.
\end{equation}
In this section, we consider
two such, proving in each case an optimality result for LMCs
constructed from SICs. To understand the conceptual significance of
the bounds we will derive, consider again Alice who is preparing to
send a quantum system through an LMC. Alice initially ascribes the
quantum state $\rho$ to her system, and before sending the system
through the channel, she computes $\mathcal{E}(\rho)$. After eliciting
a measurement outcome, Alice will update her quantum-state assignment,
not to $\mathcal{E}(\rho)$ but rather to whichever $\rho'_i$
corresponds to the outcome $E_i$ that actually transpires. The state
$\mathcal{E}(\rho)$ will generally be mixed, while $\rho'_i$ will be a
pure state in the case of a rank-1 MIC. This change from mixed to pure represents
a sharpening of Alice's expectations about her quantum system. We can
quantify this in entropic terms, even for MICs that are not rank-1. In
fact, for pure state inputs we can calculate Alice's \emph{typical sharpening of
expectations} by averaging the post-channel von Neumann entropy over the possible input states using the Haar measure, denoted $\overline{S(\mathcal{E}(\ketbra{\psi}{\psi}))}$. We will
see that SIC channels give the largest possible typical sharpening of
expectations.

In the following we make use of a partial ordering on real vectors arranged in nonincreasing order called \textit{majorization}~\cite{Marshall:2011}. A real vector $x$ rearranged into nonincreasing order is written as $x^\downarrow$. Then we say a vector $x$ majorizes a vector $y$, denoted $x\succ y$, if all of the leading partial sums of $x^\downarrow$ are greater than or equal to the leading partial sums of $y^\downarrow$ and if the sum of all the elements of each is equal. Explicitly, $x\succ y$ if
\begin{equation}
    \sum_{i=1}^kx^\downarrow_i\geq\sum_{i=1}^ky^\downarrow_i\;,
\end{equation}
for $k=1\ldots N-1$ and $\sum_ix_i=\sum_iy_i$. Speaking heuristically, if $x \succ y$, then $y$ is a flatter vector than $x$. A \textit{Schur convex} function is a function $f$ satisfying the implication $x\succ y\implies f(x)\geq f(y)$. A function is \textit{strictly} Schur convex if the inequality is strict when $x^\downarrow\neq y^\downarrow$. When the inequality is reversed the function is called \textit{Schur concave}. 

\begin{theorem}\label{avgentropybnd}
  Let $\mathcal{E}$ be an
  LMC. $\overline{S(\mathcal{E}(\ketbra{\psi}{\psi}))} \leq \log(d+1)
  - \frac{2}{d+1}\log 2$ with equality achievable if a SIC exists in
  dimension $d$.
\end{theorem}
\begin{proof}
    From Lemma \ref{lemma}, we know that the
  average maximal eigenvalue for the output of an arbitrary LMC given
  a pure state input is lower bounded by $\frac{2}{d+1}$. This implies
  \begin{equation}\label{avgmaj}
      \begin{split}
      \overline{\lambda}&\succ
    \left(\overline{\lambda}_{\rm max},
    \frac{1-\overline{\lambda}_{\rm max}}{d-1},
    \ldots,\frac{1-\overline{\lambda}_{\rm max}}{d-1}\right)\\
   & \succ \left(\frac{2}{d+1},\frac{1}{d+1},\ldots,\frac{1}{d+1}\right)\;.
   \end{split}
  \end{equation}
  The Shannon entropy $H(P)=-\sum_iP_i\log P_i$ is a concave and Schur concave function of probability distributions. Furthermore, the von Neumann entropy of a density matrix is equal to the Shannon entropy of its eigenvalue spectrum. Using these facts we have
  \begin{equation}
      \begin{split}
      \overline{S(\mathcal{E}(\ketbra{\psi}{\psi}))}&=\overline{H\big(\lambda(\mathcal{E}(\ketbra{\psi}{\psi}))\big)}\\
      &\leq H\!\left(\overline{\lambda(\mathcal{E}(\ketbra{\psi}{\psi}))}\right)\\
      &\leq \log(d+1)
  - \frac{2}{d+1}\log 2\;.
  \end{split}
  \end{equation}
If a SIC exists, taking $\mathcal{E}=\mathcal{E}_{\rm SIC}$
  achieves this upper bound.
\end{proof}
Theorem \ref{avgentropybnd} would have been more forceful if the upper bound were saturated ``only if'' a SIC exists, but we were unable to demonstrate this property, and so we leave it as a conjecture:
\begin{conjecture}\label{onlyif}
  Equality is achievable in the statement of Theorem
  \ref{avgentropybnd} only if a SIC exists in dimension $d$.
\end{conjecture}
\noindent We were, however, able to prove a strong SIC optimality result in the setting of bipartite
systems, applicable for example to Bell-test scenarios. The
\emph{entropy exchange} for a channel $\mathcal{E}$ upon input by
state $\rho$ is defined~\cite{Nielsen:2010} to be the von Neumann
entropy of the result of sending one half of a purification of $\rho$,
$\ket{\Psi_\rho}$, through the channel:
\begin{equation}
  S(\rho,\mathcal{E})
  := S\big(I\otimes\mathcal{E}(\ketbra{\Psi_\rho}{\Psi_{\rho}})\big)\;.
\end{equation}

\begin{theorem}\label{entropyexchangethm}
  Let $\mathcal{E}$ be an LMC. Then $S\!\left(\frac{1}{d}I,\mathcal{E}\right) \leq
  \log d+\frac{d-1}{d}\log(d+1)$ with equality achievable iff a SIC
  exists in dimension $d$.
\end{theorem}

\begin{proof}
    The purification of the state $\frac{1}{d}I$ is the maximally entangled state $\ket{M\!E}:=\frac{1}{\sqrt{d}}\sum_i\ket{ii}$. Let $\lambda$ be the eigenvalues of
  $I\otimes\mathcal{E}(\ketbra{M\!E}{M\!E})$ arranged in nonincreasing
  order. We may lower bound the maximal eigenvalue as follows:
\begin{equation}\label{maxbound}
  \begin{split}
    \lambda_{\rm max}
    &\geq\bra{M\!E}I\otimes \mathcal{E}(\ketbra{M\!E}{M\!E})\ket{M\!E}\\
    &=\frac{1}{d^2}\sum_{ijkl}\bra{ii}I \otimes
    \mathcal{E}(\ketbra{jj}{kk})\ket{ll}\\
    &=\frac{1}{d^2}\sum_{ijkl}\bra{ii}(\ketbra{j}{k} \otimes
    \mathcal{E}(\ketbra{j}{k}))\ket{ll}\\
    &=\frac{1}{d^2}\sum_{il}\bra{i}\mathcal{E}(\ketbra{i}{l})\ket{l}\\
    &=\frac{1}{d^2}\sum_{ijl}\bra{i}\sqrt{E_j}\ketbra{i}{l}\sqrt{E_j}\ket{l}\\
    &=\frac{1}{d^2}\sum_j\left(\tr\!\sqrt{E_j}\right)^2
    \geq \frac{1}{d^2}\sum_je_j=\frac{1}{d}\;.
  \end{split}
\end{equation}
Thus,
\begin{equation}
    \begin{split}
  \lambda &\succ \left(\lambda_{\rm max}, \frac{1-\lambda_{\rm max}}{d^2-1},
  \ldots,\frac{1-\lambda_{\rm max}}{d^2-1}\right)\\
  &\succ \left(\frac{1}{d},\frac{1}{d(d+1)},\ldots,\frac{1}{d(d+1)}\right)\;.
  \end{split}
\end{equation}
The upper bound now follows from the Schur concavity of von Neumann entropy.

If a SIC exists, it is easy to verify that  
\begin{equation}
    I\otimes\mathcal{E}_{\rm SIC}(\ketbra{M\!E}{M\!E})
  = \frac{\frac{1}{d}I\otimes I+\ketbra{M\!E}{M\!E}}{d+1}
\end{equation}
which saturates the upper bound. Von Neumann entropy is strictly Schur
concave~\cite{Bosyk:2016}, so the upper bound is saturated iff
$\lambda = \left(\frac{1}{d}, \frac{1}{d(d+1)}, \ldots,
\frac{1}{d(d+1)} \right)$. Equation \eqref{maxbound} shows that
$\ket{M\!E}$ is the maximal eigenstate and that $\{E_j\}$ is a rank-$1$
MIC in the same way as in Theorem \ref{existenceproof}. By the spectral decomposition, we may write
\begin{equation}
    I\otimes\mathcal{E}(\ketbra{M\!E}{M\!E})
  = \frac{1}{d}\ketbra{M\!E}{M\!E}+\frac{1}{d(d+1)}\sum_{i=2}^{d^2}P_i
\end{equation}
where $P_i$ are projectors into the other $d^2-1$ eigenstates. As the
full set of projectors forms a resolution of the identity, we have
\begin{equation}
  \sum_{i=2}^{d^2}P_i=I\otimes I-\ketbra{M\!E}{M\!E}\;,
\end{equation}
so 
\begin{equation}
    I\otimes\mathcal{E}(\ketbra{M\!E}{M\!E})
    = \frac{\frac{1}{d}I\otimes I+\ketbra{M\!E}{M\!E}}{d+1}\;.
\end{equation}
It follows from \eqref{qsicproperty} in \hyperref[quasiSIC]{Appendix A} that
\begin{equation}\label{MEidentity}
  \ketbra{M\!E}{M\!E} = \frac{d+1}{d^2}\sum_{i=1}^{d^2}Q^T_i\otimes
  Q_i-\frac{1}{d}I
  \otimes I\;,
\end{equation}
where the $Q_i$ are elements of a quasi-SIC. From the previous expression we now have
\begin{equation}\label{spectral}
    I\otimes\mathcal{E}(\ketbra{M\!E}{M\!E})=\frac{1}{d^2}\sum_iQ^T_i\otimes Q_i\;.
\end{equation}
Applying $I\otimes\mathcal{E}$ directly to equation \eqref{MEidentity}
gives us
\begin{equation}\label{direct}
  \begin{split}
      I\otimes\mathcal{E}(\ketbra{M\!E}{M\!E})
    &=\frac{d+1}{d^2}\sum_iQ^T_i\otimes\mathcal{E}(Q_i)-\frac{1}{d}I\otimes I\\
    &=\frac{1}{d^2}\sum_iQ^T_i\otimes[(d+1)\mathcal{E}(Q_i)-I]\;,
  \end{split}
\end{equation}
where we used that every LMC is unital and that
$\frac{1}{d}\sum_iQ^T_i=I$. Comparing equations \eqref{spectral} and
\eqref{direct}, we may see that
\begin{equation}
  Q_i=(d+1)\mathcal{E}(Q_i)-I
\end{equation}
by multiplying both sides by $\widetilde{Q}^T_j\otimes I$ and tracing
over the first subsystem. The quasi-SICs form a basis for operator
space, so it follows by linearity that
\begin{equation}
  \mathcal{E}(\rho)=\frac{I+\rho}{d+1}\;,
\end{equation}
and so by Theorem \ref{existenceproof} we are done.
\end{proof}

\section{Conclusions}

In prior works we have emphasized the importance of MICs as a special
class of measurements. The considerations of this paper developed from
the idea that MICs may naturally furnish important classes of quantum
channels as well. We affirmed this intuition with the introduction of LMCs which enabled us to discover several new ways in which SICs occupy
a position of optimality among all MICs, supposing they exist. The
appearance of additional equivalences with SIC existence plays two
important roles. First, it should aid those trying to prove the SIC
existence conjecture in all finite dimensions, and second, to our
minds, it suggests that LMCs are a more important family of quantum
channels than has been realized. We hope this work will inspire more
study of LMCs and other types of channels derived from MICs not investigated
here.

One example of such an alternative is a procedure where, when the
agent implementing the channel applies the MIC, they \emph{reprepare}
the measured system in such a way that they ascribe a fixed quantum
state to it, the choice of new state being made based on the
measurement outcome. The action of such a channel is
\begin{equation}
  \mathcal{E}(\rho) = \sum_i (\tr \rho E_i) \sigma_i,
\end{equation}
where the states $\{\sigma_i\}$ are the new preparations applied in
consequence to the measurement outcomes. Channels defined by a POVM
and a set of repreparations are known as \emph{entanglement-breaking
  channels}~\cite{Ruskai:2003}. When the POVM is a MIC, we can speak
of an \emph{entanglement-breaking MIC channel} (EBMC). EBMCs coincide
with LMCs for rank-1 MICs and repreparations proportional to the MIC, but not in general. While earlier work
already gives some indication that SIC channels are significant among
EBMCs~\cite{DeBrota:2018a}, we suspect that there is much more to be
discovered about EBMCs as a class.
\section*{Acknowledgments}
We gratefully acknowledge helpful
conversations with Gabriela Barreto Lemos, Christopher A.\ Fuchs and
Jacques Pienaar. Also we thank two anonymous referees for identifying points that required clarification and locating an error in our original proof of Lemma \ref{lemma}. This research was supported in part by the John Templeton
Foundation. The opinions expressed in this publication are those of
the authors and do not necessarily reflect the views of the John
Templeton Foundation. JBD was supported by grant FQXi-RFP-1811B of the
Foundational Questions Institute. 

\section*{Postscript}
Due to a breakdown of our university email system, it was not until after \booktitle{Physical Review A} published this article that we became aware of the preprint ``Entanglement Breaking Rank'' by Pandey \textit{et al.}~\cite{Pandey:2018}. Their Corollary 3.3 is equivalent to our Theorem 1, albeit proved from a different starting point. They consider all channels having the same action as $\mathcal{E}_{\rm SIC}$ and ask when those channels can be achieved using only $d^2$ rank-1 Kraus operators. We consider channels defined by $d^2$ Kraus operators (of arbitrary rank) and ask
when they can have the action of the SIC channel. We regret this oversight, and we commend their paper to the reader's attention. The silver lining is that we can now say the SIC problem has attracted sufficient interest that the literature is not easy to keep up with. Moreover, our attention having been called back to this paper after an interlude thinking about other aspects of SICs, we now believe that Conjecture 1 can be proven for the special case of unbiased LMCs in $d = 2$. We now sketch
the argument here.

Consider a MIC whose elements are constructed by taking the orbit of an operator under the action of a discrete group of unitaries. Such a MIC is known as \emph{group covariant} and is necessarily unbiased. If $\{\ket{\psi_j}\}$ is a set covariant with respect to the same group as the MIC, then the post-channel states $\{\mathcal{E}(\ketbra{\psi_j}{\psi_j})\}$ are unitarily equivalent and thus isospectral. We have the eigenvalue bound
\begin{equation}
    \lambda_{\rm max}(\mathcal{E}(\ketbra{\psi_0}{\psi_0})) \geq \bra{\psi_0}\mathcal{E}(\ketbra{\psi_0}{\psi_0}) \ket{\psi_0}\;. 
\end{equation}
In the case that the states $\{\ket{\psi_j}\}$ comprise a SIC and the LMC is unbiased and rank-1, that is, the MIC elements take the form $E_i=\frac{1}{d}\ketbra{\phi_i}{\phi_i}$, we can evaluate this bound by transferring the unitary transformations from the LMC elements to the SIC states:
\begin{equation}
    \lambda_{\rm max}(\mathcal{E}(\ketbra{\psi_0}{\psi_0}))
 \geq \frac{1}{d} \sum_i |\braket{\phi_0}{\psi_i}|^4 = \frac{2}{d+1}\;.
\end{equation}
Supposing the entropic bound in Theorem 2 is saturated, then the MIC must be rank-1, and the average maximum eigenvalue is equal to $2/(d+1)$. If a SIC exists, then the discrete average of the maximum eigenvalue over the SIC-state inputs is equal to the continuous average over all pure states, because a SIC is a 2-design. Therefore, if a SIC exists and the entropic bound is saturated, then the maximum eigenvalue of each post-channel state for any SIC-state input is exactly $2/(d+1)$. In addition, the eigenvector of the post-channel state corresponding to this eigenvalue is the SIC-state input itself.

Eigenvalue information is most helpful in $d = 2$. Knowing the maximum eigenvalue fixes the only other eigenvalue, and from the above, we have the complete eigendecomposition of the post-channel state for any SIC-state input. From here, we can essentially do quantum channel tomography, fixing by linearity the action of the channel.

A careful study of Bloch-sphere geometry shows that an unbiased rank-1 MIC
in $d = 2$ is necessarily group covariant, and in fact is unitarily equivalent to a MIC covariant under the Pauli group. Therefore, knowing that an LMC in $d = 2$ is unbiased and that the entropic bound is saturated, we know the MIC is group covariant, and the above argument applies.

\appendix
\section*{Appendix A}\label{quasiSIC}
Here we define and construct the quasi-SICs which furnish the Kraus operators needed in Theorem \ref{existenceproof} and which were referenced in Theorem \ref{entropyexchangethm}. Although SIC existence is not assured, one may always form a
  \emph{quasi-SIC} in any finite dimension $d$. A quasi-SIC is a set
  of Hermitian operators obeying the same Hilbert--Schmidt inner
  product condition as the SIC projectors. As positivity is not
  demanded, it is relatively easy to construct a quasi-SIC as
  follows~\cite{Appleby:2017}. Start with an orthonormal basis for the
  Lie algebra $\mathfrak{su}(d)$ of traceless Hermitian
  operators. With the Hilbert--Schmidt inner product this space is a
  $(d^2-1)$-dimensional Euclidean space, so it is possible to
  construct a regular simplex $\{B_i\}$ consisting of $d^2$ normalized
  traceless Hermitian operators. In this case $\tr B_iB_j =
  \frac{-1}{d^2-1}$ when $i\neq j$. Then the operators
  \begin{equation}
    Q_i = \sqrt{\frac{d-1}{d}}B_i+\frac{1}{d}I
  \end{equation}
  form a quasi-SIC. 

  It turns out that $A_i=\frac{1}{\sqrt{d}}Q_i$ give a set of Kraus
  operators such that
  \begin{equation}
      \mathcal{E}(\rho)=\sum_jA_j\rho A_j^\dag=\frac{I+\rho}{d+1}\;,
  \end{equation}
  or, more generally, for an arbitrary operator $X$,
  \begin{equation}\label{qsicaction}
      \mathcal{E}(X)=\frac{(\tr X)I+X}{d+1}\;,
  \end{equation}
  that is, equivalent to the action of $\mathcal{E}_{\rm SIC}$.
To see this, first observe from Corollary 1 in \cite{Appleby:2015} that
  \begin{equation}\label{qsicproperty}
      \begin{split}
    \frac{1}{d}\sum_iQ_i\otimes Q_i^T
    &= \frac{2}{d+1}P_{\rm sym}^{T_B}\\
    &= \frac{1}{d+1}\left(I^{\otimes 2}+\sum_{kl}\ketbra{kk}{ll}\right),
\end{split}
  \end{equation}
  where $T_B$ indicates the partial transpose over the second
  subsystem. Then, with the help of the vectorized notation for an operator $\kket{A}:=\sum_iA\otimes I\ket{i}\ket{i}$ and the identity
  $\kket{BAB}=B\otimes B^T\kket{A}$, we have
  \begin{equation}
      \begin{split}
    \kket{\mathcal{E}(X)} &= \frac{1}{d}\sum_iQ_i\otimes
    Q_i^T\kket{X}\\ 
    &=\frac{1}{d+1}\left(I^{\otimes 2}+\sum_{kl}\ketbra{kk}{ll}\right)\kket{X}\\
    &=\frac{1}{d+1}\left(\kket{X}+\sum_{klm}\ketbra{k}{l}X\ket{m}\otimes\ketbra{k}{l}I\ket{m}\right)\\
    &=\frac{1}{d+1}\left(\kket{X}+\sum_{km}\bra{m}X\ket{m}\ket{k}\ket{k}\right)\\
    &=\frac{1}{d+1}\left(\kket{X}+\kket{(\tr X)I}\right)\\
    &= \biggr\lvert\frac{(\tr X)I+X}{d+1}\biggr\rangle\!\!\!\biggr\rangle\;,
\end{split}
  \end{equation}
  which is equivalent to \eqref{qsicaction}. Sending $X=I$ through equation \eqref{qsicaction} reveals the identity 
  \begin{equation}\label{qsicsquaresum}
      \frac{1}{d}\sum_iQ_i^2=I\;,
  \end{equation}
  which, since the quasi-SICs are Hermitian, is equivalent to the requirement that Kraus operators satisfy $\sum_iA_i^\dag A_i=I$.

\section*{Appendix B}\label{dblyquasi}
A doubly quasistochastic matrix is a matrix of real numbers whose rows and columns sum to $1$. If we assume that $\{E_j\}$
is an unbiased MIC, $E_i=\frac{1}{d}\rho_i$, we will now show that
$U$ is furthermore doubly quasistochastic.

The Gram matrix for the $K_i$ operators \eqref{quasikraus} is
\begin{equation}
  \tr K_iK_j=(1/d-\gamma)\delta_{ij}+\gamma 
\end{equation}
where 
\begin{equation}
  \gamma=\frac{d-1-(d-2)\beta+2\sqrt{(1-\beta)(1-\beta+d\beta)}}{d(d+1)}\;.
\end{equation}
Since $e_i=1/d=\tr E_i=\tr\sqrt{E_i}\sqrt{E_i}$, we have
\begin{equation}
  \begin{split}
    \frac{1}{d}&=\sum_{jk}[U]_{ij}[U]_{ik}\tr K_jK_k\\
    &=\sum_{jk}[U]_{ij}[U]_{ik}\big((1/d-\gamma)\delta_{jk}+\gamma\big)\\
    &=(1/d-\gamma)\sum_{jk}[U]_{ij}[U]_{ik}\delta_{jk}
    + \gamma\left(\sum_j[U]_{ij}\right)^{\!2}\\
    &=(1/d-\gamma)+\gamma\left(\sum_j[U]_{ij}\right)^{\!2}\;,
  \end{split}
\end{equation}
from which we obtain
\begin{equation}\label{usum2}
    \sum_j[U]_{ij}=1\;.
\end{equation}
Now note that 
\begin{equation}
    \tr K_i=\frac{\sqrt{1-\beta+d\beta}+(d-1)\sqrt{1-\beta}}{d}\;.
\end{equation}
Tracing both sides of \eqref{remix2} reveals that $\tr\sqrt{E_i}=\tr
K_i$ is a constant. Corollary 3 from \cite{Appleby:2015} then asserts that
\begin{equation}
    \begin{split}
  \sum_i\sqrt{E_i} &= \sum_j\left(\sum_i[U]_{ij}\right)K_j\\
  &= \sqrt{d(1/d-\gamma)+d^3\gamma}I\;.
  \end{split}
\end{equation}
Summing equation \eqref{quasikraus} gives
\begin{equation}
  \sum_iK_i=\left(\sqrt{1-\beta+d\beta}+(d-1)\sqrt{1-\beta}\right)I\;.
\end{equation}
As the $K_j$ form a linearly independent set, combining the previous
two equations requires
\begin{equation}
  \sum_i[U]_{ij}=\frac{\sqrt{1-d\gamma+d^3\gamma}}{\sqrt{1-\beta+d\beta}
    + (d-1)\sqrt{1-\beta}}
  = 1\;.
\end{equation}
Thus $U$ is doubly quasistochastic, as claimed.

\begin{thebibliography}{999}

\bibitem{DeBrota:2018a} J.\ B.\ DeBrota, C.\ A.\ Fuchs and
    B.\ C.\ Stacey, ``Symmetric informationally complete measurements identify the irreducible difference between classical and quantum systems,'' \hrefdoi{10.1103/PhysRevResearch.2.013074}{\booktitle{Physical Review Research} \textbf{2} (2020), 013074},
  \arxiv{1805.08721}.

\bibitem{DeBrota:2018b} J.\ B.\ DeBrota, C.\ A.\ Fuchs and
  B.\ C.\ Stacey, ``Analysis and Synthesis of Minimal Informationally Complete Quantum Measurements,'' \arxiv{1812.08762} (2018).
 
\bibitem{Zauner:1999} G.\ Zauner, \booktitle{Quantendesigns.\ {G}rundz{\"u}ge einer\\ nichtkommutativen
    {D}esign\-theorie}. PhD Thesis,\\ University of Vienna (1999),\\
  \url{http://www.gerhardzauner.at/qdmye.html}.

\bibitem{Renes:2004} J.\ M.\ Renes, R.\ Blume-Kohout,
  A.\ J.\ Scott and C.\ M.\ Caves,
  ``\hrefdoi{10.1063/1.1737053}{Symmetric Informationally Complete
    quantum measurements},'' \booktitle{Journal of Mathematical
    Physics} \textbf{45} (2004), 2171--80, \arxiv{quant-ph/0310075}.

\bibitem{Scott:2010} A.\ J.\ Scott and M.\ Grassl,
  ``\hrefdoi{10.1063/1.3374022}{Symmetric informationally complete
  positive-operator-valued measures:\ A new computer study},''
  \booktitle{Journal of Mathematical Physics} \textbf{51} (2010),
  042203, \arxiv{0910.5784}.

\bibitem{Fuchs:2017} C.\ A.\ Fuchs, M.\ C.\ Hoang and B.\ C.\ Stacey,
  ``\hrefdoi{10.3390/axioms6030021}{The SIC question:\ History and
  state of play},'' \booktitle{Axioms} \textbf{6} (2017), 21,
  \arxiv{1703.07901}.

\bibitem{Nielsen:2010} M.\ A.\ Nielsen and I.\ Chuang,
  \booktitle{Quantum Computation and Quantum Information} (Cambridge
  University Press, 2010).

\bibitem{Barnum:2000} H.\ Barnum, ``Information-disturbance tradeoff
  in quantum measurement on the uniform ensemble and on the mutually
  unbiased bases,'' \arxiv{quant-ph/0205155} (2000).

\bibitem{Busch:2009} P.\ Busch and P.\ Lahti,
  ``\hrefdoi{10.1007/978-3-540-70626-7_110}{L\"uders Rule}.'' In
  \booktitle{Compendium of Quantum Physics} (Springer,
  2009). \url{http://philsci-archive.pitt.edu/4111/}.
  
\bibitem{Fuchs:2012} C.\ A.\ Fuchs and R.\ Schack,
  ``\hrefdoi{10.1007/978-3-642-21329-8_15}{Bayesian conditioning, the
  reflection principle, and quantum decoherence}.'' In
  \booktitle{Probability in Physics} (Springer,
  2012). \arxiv{1103.5950}.

\bibitem{Stacey:2019} B.\ C.\ Stacey, ``Quantum theory as symmetry
  broken by vitality,'' \arxiv{1907.02432} (2019).

\bibitem{Bengtsson:2017} I.\ Bengtsson and K.\ \.{Z}yczkowski,
  \booktitle{Geometry of Quantum States:\ An Introduction to Quantum
    Entanglement}, second edition (Cambridge University Press, 2017).
  
\bibitem{Appleby:2017} M.\ Appleby, C.\ A.\ Fuchs, B.\ C.\ Stacey and
  H.\ Zhu, ``\hrefdoi{10.1140/epjd/e2017-80024-y}{Introducing the
    Qplex:\ A novel arena for quantum theory},'' \booktitle{European
    Physical Journal D} \textbf{71} (2017), 197, \arxiv{1612.03234}.

\bibitem{Appleby:2015} M.\ Appleby, C.\ A.\ Fuchs and H.\ Zhu,
  ``Group theoretic, Lie algebraic and Jordan algebraic
  formulations of the SIC existence problem,'' \booktitle{Quantum
    Information \& Computation} \textbf{15} (2015), 61--94,
  \arxiv{1312.0555}.

\bibitem{Zhu:2016} H.\ Zhu,
  ``\hrefdoi{10.1103/PhysRevLett.117.120404}{Quasiprobability
  representations of quantum mechanics with minimal negativity},''
  \booktitle{Physical Review Letters} \textbf{117} (2016), 120404,
  \arxiv{1604.06974}.
  
\bibitem{DeBrota:2017} J.\ B.\ DeBrota and C.\ A.\ Fuchs, ``Negativity
  bounds for Weyl--Heisenberg quasiprobability representations,''
  \booktitle{Foundations of Physics} \textbf{47} (2017), 1009--30,
  \arxiv{1703.08272}.

\bibitem{Marshall:2011} A.\ W.\ Marshall, I.\ Olkin and
  B.\ C.\ Arnold, \booktitle{Inequalities: Theory of Majorization and
    Its Applications} (Springer, 2011).

    
\bibitem{Bosyk:2016} G.\ M.\ Bosyk, S.\ Zozor, F.\ Holik, M.\ Portesi
  and P.\ W.\ Lamberti, ``\hrefdoi{10.1007/s11128-016-1329-5}{A family
    of generalized quantum entropies:\ definition and properties},''
  \booktitle{Quantum Information Processing} \textbf{15} (2016),
  3393--420, \arxiv{1506.02090}.

\bibitem{Ruskai:2003} M.\ B.\ Ruskai,
  ``\hrefdoi{10.1142/S0129055X03001710}{Qubit entanglement breaking
  channels},'' \booktitle{Reviews of Mathematical Physics} \textbf{15}
  (2003), 643--62, \arxiv{quant-ph/0302032}.

\bibitem{Wilde:2017} M.\ M.\ Wilde,
    \booktitle{Quantum Information Theory} (Cambridge University Press, 2017), \arxiv{1106.1445}.

 \bibitem{Pandey:2018} S.~K. Pandey, V.~I. Paulsen, J. Prakash, and M. Rahaman, ``Entanglement Breaking Rank,'' \arxiv{1805.04583} (2018).
\end{thebibliography}
\end{document}